\documentclass[prl,twocolumn]{revtex4}

\usepackage{amsmath}
\usepackage{graphicx}
\usepackage{amsthm}
\usepackage{enumitem}

\theoremstyle{remark}

\begin{document}

\newtheorem{prop}{Proposition}

\title{Monogamy inequality for any local 
quantum resource and entanglement}
\author{S. Camalet}
\affiliation{Laboratoire de Physique Th\'eorique 
de la Mati\`ere Condens\'ee, UMR 7600, Sorbonne 
Universit\'es, UPMC Univ Paris 06, F-75005, 
Paris, France}

\begin{abstract}
We derive a monogamy inequality for any local 
quantum resource and entanglement. It results from 
the fact that there is always a convex measure for 
a quantum resource, as shown here, and from 
the relation between entanglement and local entropy. 
One of its consequences is an entanglement 
monogamy different from that  usually discussed. 
If the local resource is nonuniformity or coherence, 
it is satisfied by familiar resource and entanglement 
measures. The ensuing upper bound for the local 
coherence, determined by the entanglement, 
is independent of the basis used to define the coherence.
\end{abstract} 

\maketitle 

The more a two-level system is quantum-mechanically 
entangled with another two-level system, the less it 
can be entangled with a third one \cite{CKW}. 
This behavior, known as entanglement monogamy, 
has also been found for larger systems 
\cite{KW,OV,BXW}. Quantitatively, it is described by 
an inequality, involving a bipartite entanglement 
monotone. The term ''monotone'' refers to the fact 
that a proper measure of entanglement, cannot 
increase under specific transformations of quantum 
states. They are those that can be achieved by local 
operations and classical communication, and hence, 
cannot generate entanglement \cite{V,PV,HHHH}. 

Quantum entanglement theory is a quantum resource 
theory. This approach to quantum resources, is based 
on the definition of free states, for which the resource 
vanishes, and free operations. The set of 
these allowed operations, depends on the considered 
theory \cite{CG}. But, in any case, they transform 
free states into free states \cite{BG}. For entanglement, 
the free states are the so-called separable states 
\cite{W}, and the free operations are those obtained 
from local operations and classical communication. 
Other examples of quantum resources are 
nonuniformity \cite{HHHHOSS,HHO,GMNSH}, 
athermality \cite{JWZGB,A,HO,BHORS}, 
and coherence \cite{CG,BCP,SSDBA,CH}. 
Contrary to entanglement, the definitions of 
these resources, do not rely on a partition of 
the system of interest. Similarly to entanglement 
monotones, a measure of a given resource, cannot 
increase under the corresponding free operations. 

Recently, a monogamy inequality for entanglement 
and local contextuality, has been derived 
\cite{PRA1,aX}. It shows that the entanglement 
between two systems, and the violation, by one 
of them, of a state-dependent noncontextuality 
inequality, constrain each other. In the context of open 
sytems, the development of the entanglement between 
a system and its environment, due to their mutual 
interaction, is commonly seen as playing an essential 
role in the decoherence of the system \cite{PRA2,S}. 
These influences on the quantum properties 
of a system, of its entanglement with another system, 
suggest that there may be a general monogamy 
relation between any local quantum resource and 
entanglement. The interplay between entanglement 
and local coherence, has also been discussed from 
other perspectives \cite{SSDBA,CH}.

In this Letter, we derive a monogamy inequality for 
any local quantum resource and entanglement. 
This inequality involves a convex measure of 
the resource. Though, for resource measures, 
convexity is frequently assumed \cite{BCP}, it is not 
a basic axiom \cite{HHHH,SSDBA}. However, we show 
that, as soon as there exists a proper measure for 
a resource, i.e., that does not increase under free 
operations, there is also a convex one. The derivation 
of the monogamy inequality, essentially relies on 
the relation between entanglement and local entropy 
\cite{aX}. By entropy, we mean a non-negative 
function of quantum states, that depends only on 
the state eigenvalues, is non-decreasing with 
the state mixedness, in the sense of majorization 
\cite{MOA}, and vanishes for pure states 
\cite{K,BZHPL}. The found monogamy inequality is 
discussed for three local resources: entanglement, 
nonuniformity, and coherence.

Two essential requirements for a resource measure, 
are that (i) it is non-negative and vanishes for free 
states, (ii) it is non-increasing under free operations. 
Such a measure consists of a set of functions $R_d$, 
from $d\times d$ density matrices to real numbers 
\cite{GMNSH}, or, more generally, of a set of 
functions $R_{\boldsymbol d}$, where 
$\boldsymbol d \equiv (d_1,d_2,\ldots)$, 
from density matrices on the Hilbert space 
${\cal H}_{d_1}\otimes {\cal H}_{d_2} 
\otimes \ldots$, to real numbers. Condition (i) simply 
means that $R_{\boldsymbol d} \ge 0$, and 
$R_{\boldsymbol d}(\rho_\mathrm{A})=0$ when 
$\rho_\mathrm{A}$ is free, for any $\boldsymbol d$. 
We denote $\rho_\mathrm{A}$ the state of 
the system of interest, possibly composite, named A, 
since a bipartite system, consisting of A and another 
system, is considered in the following. Note that 
$R_{\boldsymbol d}$ can vanish for some resourceful 
states. The important point is that it is zero for 
all free states \cite{BCP,HHHH}. The monotonicity 
condition (ii) reads, more precisely, as
\begin{equation}
R_{\boldsymbol d'}
\left(\sum_q K^{\phantom \dag}_q 
\rho_\mathrm{A} K^\dag_q \right)
\le R_{\boldsymbol d}(\rho_\mathrm{A}) , 
\label{mio1} 
\end{equation}
where $K_q$ are the Kraus operators of 
the considered free operation, which are such that 
$\sum_q K^\dag_qK^{\phantom \dag}_q$ 
is equal to the identity operator on 
${\cal H}_{\boldsymbol d} \equiv 
{\cal H}_{d_1}\otimes {\cal H}_{d_2} 
\otimes \ldots$. It can also be understood as
\begin{equation}
\sum_q p_q R_{\boldsymbol d_q} 
\left(\rho_\mathrm{A}^{(q)}\right)
\le R_{\boldsymbol d}(\rho_\mathrm{A}), 
\label{mio2}
\end{equation}
where $p_q=\operatorname{tr} 
(K^\dag_qK^{\phantom\dag}_q 
\rho_\mathrm{A})$ is the probability of outcome 
$q$, $\rho_\mathrm{A}^{(q)}=
K^{\phantom\dag}_q \rho_\mathrm{A} 
K^\dag_q/p_q$ is the corresponding state, and 
the sum runs over $q$ such that $p_q>0$ 
\cite{BCP}. In eq.\eqref{mio1}, all the linear 
maps $K_q$ are from ${\cal H}_{\boldsymbol d}$ 
to ${\cal H}_{\boldsymbol d'}$, with the same 
$\boldsymbol d'$, whereas, in eq.\eqref{mio2}, 
the vectors $\boldsymbol d_q$ can be different from 
one another. As shown in the supplemental material 
\cite{SM}, if there are functions $R_{\boldsymbol d}$ 
satisfying the above points (i) and (ii), then there are 
convex ones, $R^{ch}_{\boldsymbol d}$, which obey 
the same conditions. They fulfill (ii) in the same way 
as the set $\{ R_{\boldsymbol d} \}$ does. Besides, 
$\{ R^{ch}_{\boldsymbol d} \}$ always obeys 
eq.\eqref{mio1}, since the functions 
$R^{ch}_{\boldsymbol d}$ are convex \cite{BCP}. 
As $R^{ch}_{\boldsymbol d} \le R_{\boldsymbol d}$, 
see supplemental material, it vanishes whenever 
$R_{\boldsymbol d}$ does, but it can be zero for 
other states. If, for example, the set of free states 
is not convex, there are resourceful states which are 
convex combinations of free states. For these states, 
$R^{ch}_{\boldsymbol d}$ necessarily vanishes, 
but $R_{\boldsymbol d}$ may not.

Let us now consider an arbitrary convex function 
$R_{\boldsymbol d}$, of the density operators on 
${\cal H}_{\boldsymbol d}$, non-negative and 
bounded. From $R_{\boldsymbol d}$, we define, 
for any state $\rho_\mathrm{A}$ of any finite 
system A,
\begin{equation}
G(\rho_\mathrm{A}) \equiv 
\sup_{\{ | i \rangle \}} R_{\boldsymbol d}
\left( \sum_{i=1}^r \lambda_i(\rho_\mathrm{A}) 
| i \rangle \langle i | \right) , \label{K}
\end{equation}
for $r \le n(\boldsymbol d)$, and $0$ otherwise, 
where $n(\boldsymbol d) \equiv d_1 d_2 \ldots$ 
is the dimension of ${\cal H}_{\boldsymbol d}$, 
$r$ is the rank of $\rho_\mathrm{A}$, and 
$\lambda_i(M)$ denotes the eigenvalues of 
the Hermitian operator $M$, with 
$\lambda_i(M) \ge \lambda_{i+1}(M)$.
The supremum is taken over the bases 
$\{ | i \rangle \}$ of ${\cal H}_{\boldsymbol d}$. 
The function \eqref{K} depends on 
$\rho_\mathrm{A}$ only via the nonvanishing 
eigenvalues $\lambda_i(\rho_\mathrm{A})$. 
Contrary to $R_{\boldsymbol d}$, it is defined 
for any state, and does not depend on 
the corresponding Hilbert space. To make 
this distinction clear, we do not denote 
its dependence on $\boldsymbol d$, which simply 
comes from the definition \eqref{K}. If the Hilbert 
space of A is ${\cal H}_{\boldsymbol d}$, 
eq.\eqref{K} reduces to $G(\rho_\mathrm{A}) 
= \sup_U R_{\boldsymbol d}
(U \rho_\mathrm{A} U^\dag)$, where 
the supremum is taken over the unitary 
operators $U$ of ${\cal H}_{\boldsymbol d}$, 
and hence, 
$R_{\boldsymbol d}( \rho_\mathrm{A} ) \le 
G(\rho_\mathrm{A})$. 

It follows, from the properties of G, shown in 
the supplemental material, that 
$R_{\boldsymbol d}^{sup}\equiv 
G(| i \rangle \langle i |)$, where $| i \rangle$ is 
any pure state, is the supremum of 
$R_{\boldsymbol d}$, and that the function $S$, 
defined by 
\begin{equation}
S(\rho_\mathrm{A}) 
\equiv R_{\boldsymbol d}^{sup}
-G(\rho_\mathrm{A}) ,\label{S}
\end{equation} 
is non-negative. Furthermore, $-S$ is Schur-convex, 
and, by construction, $S$ vanishes 
when $\rho_\mathrm{A}$ is pure. Thus, $S$ is 
an entropy, and can obey, with 
an entanglement monotone $E$,
\begin{equation}
S(\rho_\mathrm{A})= 
\max_{\rho : \operatorname{tr}_\mathrm{B} \rho 
= \rho_\mathrm{A}} E(\rho) ,
\label{SE}
\end{equation}
where the maximum is taken over the states $\rho$ of 
the composite systems, consisting of A, and another 
system, say B, such that 
$\operatorname{tr}_\mathrm{B} \rho 
= \rho_\mathrm{A}$, 
and $\operatorname{tr}_\mathrm{B}$ 
denotes the partial trace over B \cite{aX}. 
The maximum is reached for pure states $\rho$. 
Note that an entanglement monotone does not 
depend explicitly on the Hilbert space dimensions of 
A an B \cite{HHHH,RBCHM}. Since the function 
\eqref{S} is concave, an explicit entanglement 
monotone, that fulfills eq.\eqref{SE}, with $S$, can 
be built. It is the convex roof
\begin{equation}
E^{cr}(\rho)\equiv 
\inf_{\{P_k,| \Psi_k \rangle \} } \sum_k P_k S
\left(\operatorname{tr}_\mathrm{B} 
| \Psi_k \rangle \langle \Psi_k | \right) , \label{E}
\end{equation} 
where the infinum is taken over the ensembles 
$\{P_k,| \Psi_k \rangle \}$ such that $\sum_k P_k 
| \Psi_k \rangle \langle \Psi_k |=\rho$ \cite{V,aX}.
It is clearly lower than $S(\rho_\mathrm{A})$, 
as $S$ is concave. Expression \eqref{E} with $S$ 
replaced by the von Neumann entropy, is 
the definition of the entanglement of formation 
\cite{HHHH}.

Equation \eqref{SE} gives the monogamy 
inequality
\begin{equation}
R_{\boldsymbol d}(\rho_\mathrm{A})
+E(\rho) \le R_{\boldsymbol d}^{sup} ,
\label{mi}
\end{equation} 
when the Hilbert space of system A is 
${\cal H}_{\boldsymbol d}$. The entanglement 
monotone \eqref{E} satisfies this inequality, but 
it may not be the only one. As soon as 
a monotone $E$ obeys eq.\eqref{SE} with $S$,
 it fulfills eq.\eqref{mi}. Moreover, for such 
an entanglement measure, there are, 
for any local eigenspectrum 
$\{ \lambda_i(\rho_\mathrm{A}) 
\}_{i=1}^{n(\boldsymbol d)}$, 
global states $\rho$ such that the left side 
of eq.\eqref{mi}, is as close as we wish to 
$R_{\boldsymbol d}^{sup}$, 
as shown in the supplemental material. As mentioned 
above, for any quantum resource, there is a convex 
measure of it, and hence, a monogamy inequality 
\eqref{mi} for the entanglement between A and B, 
and this resource for A. For the set of states $\rho$ 
such that the two sides of eq.\eqref{mi} are equal, 
or infinitely close to each other, an increase of 
the entanglement $E(\rho)$, means a reduction of 
the local resource
$R_{\boldsymbol d}(\rho_\mathrm{A})$, of 
the same amount, and reciprocally. In general, 
entanglement and local resource limit each other. 
We remark that the monogamy inequality for 
entanglement and local contextuality, derived 
in Ref.\cite{aX}, is a particular case of eq.\eqref{mi}, 
with $R_{\boldsymbol d}$ replaced by the convex 
function $C_d$ \cite{PRA1}.

The entanglement between A and B, is not changed 
by a unitary transformation $U$ performed on A, 
and hence, $\rho_\mathrm{A}$ can be replaced 
by $U\rho_\mathrm{A}U^\dag$, in eq.\eqref{mi}.
Thus, $R_{\boldsymbol d}^{sup}-E(\rho)$ not only 
upperbounds $R_{\boldsymbol d}(\rho_\mathrm{A})$, 
but also all the values of $R_{\boldsymbol d}$ that can 
be obtained by performing local unitary transformations 
on A. This bound can be reached, in this way, 
when $\rho$ is pure. For the entanglement between 
A and B, inequality \eqref{mi} gives an upper bound, 
$R_{\boldsymbol d}^{sup}
-\sup_U R_{\boldsymbol d}(U \rho_\mathrm{A} U^\dag)$, 
where the supremum is taken over the unitary 
operators $U$ of A, that depends only on 
the eigenvalues of the local state $ \rho_\mathrm{A}$. 
We remark that the supremum of $R_{\boldsymbol d}$, 
was obtained above as 
$R_{\boldsymbol d}^{sup}=\sup_{| i \rangle} 
R_{\boldsymbol d}( | i \rangle \langle i |)$, where 
the supremum is taken over the pure states 
$| i \rangle$ of ${\cal H}_{\boldsymbol d}$. 
The convexity of $R_{\boldsymbol d}$, implies thus 
that it is equal, or infinitely close to its supremum, 
for some pure states. However, the above results do 
not impose that there exist such states independent 
of the resource measure. As is well known, this is 
the case for entanglement and coherence 
\cite{HHHH,BCP}.

An interesting particular case is when the system A 
is made up of two subsystems, $\mathrm{A}_1$ 
and $\mathrm{A}_2$, and the considered 
resource for A, is the entanglement between 
$\mathrm{A}_1$ and $\mathrm{A}_2$. Then, 
inequality \eqref{mi} can be rewritten, in a more familiar 
form, as
\begin{equation}
\tilde E (\mathrm{A}_1:\mathrm{A}_2)
+E(\mathrm{A}_1\mathrm{A}_2:B) 
\le \tilde E_{max}
\label{mei}
\end{equation} 
where the entanglement 
$\tilde E(\mathrm{A}_1:\mathrm{A}_2)$ between 
$\mathrm{A}_1$ and $\mathrm{A}_2$, and 
the entanglement $E(\mathrm{A}_1\mathrm{A}_2:B)$ 
between A and B, are evaluated for the common 
state $\rho$ of $\mathrm{A}_1$, $\mathrm{A}_2$, 
and B. The right side, $\tilde E_{max}$, is 
the maximum value of 
$\tilde E(\mathrm{A}_1:\mathrm{A}_2)$, 
reached when $\rho_\mathrm{A}
=\operatorname{tr}_\mathrm{B} \rho$ 
is a maximally entangled state of 
$\mathrm{A}_1$ and $\mathrm{A}_2$. 
Note that $E(\mathrm{A}_1\mathrm{A}_2:B)$ 
is also bounded by $\tilde E_{max}$. It attains this value 
for pure states $\rho=|\psi \rangle \langle \psi |$, 
where $|\psi \rangle$ has Schmidt coefficients 
such that $\rho_\mathrm{A}$ is absolutely 
separable \cite{KZ,VADM}. The entanglement monogamy 
described by inequality \eqref{mei}, is different from that 
usually discussed \cite{CKW,KW,OV,BXW,HHHH}. 
Equation \eqref{mei} shows that the entanglement 
between two parts of a system, and the entanglement 
of this system with another one, limit each other. 
In the extreme case, usually used to illustrate 
entanglement monogamy, of two maximally entangled 
systems $\mathrm{A}_1$ and $\mathrm{A}_2$, 
it gives $E(\mathrm{A}_1\mathrm{A}_2:B)=0$, 
as expected, since system A is in a pure state, and 
hence not correlated to any other one. 

An inequality, similar to eq.\eqref{mei}, but involving 
only one entanglement monotone, the negativity 
$E_{\cal N}$ \cite{HHHH,ZHSL,ViWe}, can be obtained. 
Any measure of the form \eqref{E}, is larger 
than $g[E_{\cal N}(\rho)]$, where $g$ is 
a non-decreasing function, given 
by $g(x)=co(h)(2x+1)$, with $co(h)$ the convex 
hull of $h$, defined, on $[1,d^*]$, by 
$h(y)\equiv \inf_{\{p_i \} \in {\cal F}(y)} 
S( \sum_{i=1}^{d^*} p_i 
| i \rangle \langle i | )$. 
In this expression, $d^*$ is the smallest of the Hilbert 
space dimensions of A and B,
$\{ | i \rangle \}_{i=1}^{d^*}$ is an arbitrary basis, 
and ${\cal F}(y)$ is the set of $\{p_i \}_{i=1}^{d^*}$ 
such that $p_i \ge 0$, $\sum_{i=1}^{d^*} p_i=1$, and 
$\sum_{i=1}^{d^*} \sqrt{p_i}=\sqrt{y}$ \cite{aX,CAF}. 
Thus, for $\tilde E=E_{\cal N}$, eq.\eqref{mei} leads to
$E_{\cal N} (\mathrm{A}_1:\mathrm{A}_2)+g
\left[ E_{\cal N}(\mathrm{A}_1\mathrm{A}_2:B) \right]
\le E_{{\cal N},max}$, 
where $E_{{\cal N},max}$ is the maximum value of 
$E_{\cal N} (\mathrm{A}_1:\mathrm{A}_2)$. 
The second term on the left side can also reach 
this value. This can be seen as follows. For a maximally 
entangled state of A and B, 
$E_{\cal N}(\mathrm{A}:\mathrm{B})=(d^*-1)/2$  
\cite{ViWe}. Since $h$, and its convex hull, are defined 
on the finite interval $[1,d^*]$, $g(d^*/2-1/2)=h(d^*)$. 
The only element of ${\cal F}(d^*)$ is $p_i=1/d^*$. 
Thus, if the Hilbert space dimension of A is not larger 
than that of B, $h(d^*)$ is determined by the maximally 
mixed state of A, which is invariant under unitary 
transformations, and not entangled, and hence, using 
eq.\eqref{S}, $h(d^*)=E_{{\cal N},max}$. The function 
$g$ defined above, depends on the systems considered. 
For a system A consisting of two two-level systems, it can 
be evaluated explicitly, see supplemental material. 

Inequality \eqref{mei} can be generalised to more 
than three systems. The usual inequality for 
entanglement monogamy, for $N$ systems 
$\mathrm{A}_k$, reads
\begin{equation}
\sum_{k \ge 2} \tilde E (\mathrm{A}_1:\mathrm{A}_k) 
\le \tilde 
E(\mathrm{A}_1:\mathrm{A}_2\ldots \mathrm{A}_N) . 
\label{uem}
\end{equation} 
It has been derived for two-level systems and different 
entanglement monotones \cite{CKW,OV,BXW}, and 
for systems of any sizes and squashed 
entanglement \cite{KW}. For a system A consisting 
of $N$ subsystems, inequalities \eqref{mei} 
and \eqref{uem} lead directly to
\begin{equation}
\sum_{k \ge 2} \tilde E (\mathrm{A}_1:\mathrm{A}_k) 
+ E(\mathrm{A}_1\ldots \mathrm{A}_N:B)
\le \tilde E_{max} , \nonumber
\end{equation} 
which relates the entanglement of one subsystem 
of A, with each of the others, and the entanglement 
between A and B. Other inequalities, for 
the entanglement between parts of a system, and 
its entanglement with another system, can be derived 
from eq.\eqref{mei} and eq.\eqref{uem}. For 
example, for three subsystems, they give
\begin{equation}
\frac{2}{3}\sum_{1 \le k < l \le 3} 
\tilde E (\mathrm{A}_k:\mathrm{A}_l) 
+ E(\mathrm{A}_1 \mathrm{A}_2\mathrm{A}_3:B)
\le \tilde E_{max} . \nonumber
\end{equation} 
To compare inequalities \eqref{mei} and 
\eqref{uem}, we consider three identical systems 
$\mathrm{A}_k$ in a permutation-symmetric state. 
For such a state, the entanglement between 
any two systems, is equal to that between 
$\mathrm{A}_1$ and $\mathrm{A}_2$, and 
the entanglement between a system and the two 
other ones, is equal to that between $\mathrm{A}_1$ 
and the composite system 
$\mathrm{A}_2\mathrm{A}_3$, consisting 
of $\mathrm{A}_2$ and $\mathrm{A}_3$. 
Inequalities \eqref{mei} and \eqref{uem} yield, 
respectively, 
$\tilde E (\mathrm{A}_1:\mathrm{A}_2)
\le \tilde E_{max}
-E(\mathrm{A}_1:\mathrm{A}_2\mathrm{A}_3)$, 
and
$\tilde E (\mathrm{A}_1:\mathrm{A}_2) \le
\tilde E(\mathrm{A}_1:\mathrm{A}_2\mathrm{A}_3)/2$. 
The above first inequality shows that the two-system 
entanglement, and the entanglement of the bipartitions 
of the global system, constrain each other, but the second 
one does not. Assume, for instance, the state of the global 
system is pure. In this case, 
$\tilde E(\mathrm{A}_1:\mathrm{A}_2\mathrm{A}_3)$, 
and $E(\mathrm{A}_1:\mathrm{A}_2\mathrm{A}_3)$, 
are given by the corresponding entropies, evaluated 
for $\rho_{\mathrm{A}_1}$. As this state depends only 
on one probability, 
$E(\mathrm{A}_1:\mathrm{A}_2\mathrm{A}_3)$, 
for $\tilde E=E_{\cal N}$, can be written in terms 
of $E_{\cal N}(\mathrm{A}_1:
\mathrm{A}_2\mathrm{A}_3)$, 
using the results of Ref.\cite{VADM}. Moreover, 
for three two-level systems in a pure state, 
inequality \eqref{uem} holds for $E^2_{\cal N}$ 
\cite{OF}. Finally, we find
\begin{equation}
E_1\le  \min \Big\{ \sqrt{2}E_2 , \sqrt{1-2E_2^2}
+\sqrt{1/4-E_2^2}-1/2\Big\}/2 , \nonumber 
\end{equation}
where 
$E_1=E_{\cal N} (\mathrm{A}_1:\mathrm{A}_2)$, 
and $E_2=E_{\cal N}(\mathrm{A}_1:
\mathrm{A}_2\mathrm{A}_3)$. 
The first term on the right side, comes from 
eq.\eqref{uem}, and increases with $E_2$, whereas 
the second one comes from eq.\eqref{mei}, and 
decreases with $E_2$. They are equal for 
$E_2 \simeq 0.415$. Thus, for $E_2 \ge 0.416$, 
the relation between $E_1$ and $E_2$, is better 
described by inequality \eqref{mei}. 
With the entanglement of formation in place 
of $E_{\cal N}$, the situation is similar, but 
$E(\mathrm{A}_1:\mathrm{A}_2\mathrm{A}_3)$ 
cannot be expressed explicitly. 

As a second example of quantum resource, 
we consider nonuniformity. From the von Neumann 
entropy $S_{vN}$, a measure $R_d$ of this resource, 
can be defined by $R_d(\rho_\mathrm{A}) = \log d 
-S_{vN}(\rho_\mathrm{A})$, where $d$ is the Hilbert 
space dimension of A \cite{GMNSH}. 
The same expression with a R\'enyi 
entropy $S_R$, of positive order, in place of $S_{vN}$, 
gives also a nonuniformity monotone \cite{GMNSH}. 
Such measures depend on $\rho_\mathrm{A} $ only 
via the eigenvalues $\lambda_i (\rho_\mathrm{A})$. 
Thus, eq.\eqref{K} and eq.\eqref{S} yield 
$S(\rho_\mathrm{A})=S_{vN} (\rho_\mathrm{A})$ 
if the rank of $\rho_\mathrm{A}$ is not greater 
than $d$, and $S(\rho_\mathrm{A})=\log d$ 
otherwise, and similarly for $S_R$. However, here, 
entanglement monotones that fulfill eq.\eqref{SE} 
with $S_{vN}$ or $S_R$, instead of $S$, are more 
useful. Such an entanglement measure $E$ satisfies 
the monogamy inequality \eqref{mi} with $R_d$, 
for any value of $d$. This directly follows 
from $E(\rho)\le S_{vN/R}(\rho_\mathrm{A})$ 
and $R_d^{sup}=\log d$. Well-known entanglement 
monotones obey eq.\eqref{SE} with $S_{vN}$, namely, 
distillable entanglement, entanglement cost, 
entanglement of formation, and relative entropy 
of entanglement \cite{HHHH,BBPS,VP}. Note that 
$-S_R$ is not necessarily convex, depending on its order, 
but is always Schur-convex \cite{BZHPL}, and can thus 
satisfy eq.\eqref{SE}. For a Tsallis entropy $S_T$, of 
positive order $q$, the procedure followed in 
ref.\cite{GMNSH}, to construct nonuniformity measures, 
leads to $R_d(\rho_\mathrm{A}) 
=R_d^{sup}-d^{q-1}S_T(\rho_\mathrm{A})$, 
where $R_d^{sup}=(d^{q-1}-1)/(q-1)$. 
An entanglement monotone $E$ that fulfills 
eq.\eqref{SE} with $S_T$, obeys, for any $d$, 
$R_d(\rho_\mathrm{A})
+d^{q-1}E(\rho) \le R_d^{sup}$, which is also 
a monogamy inequality for entanglement and 
local nonuniformity.  

We now turn to quantum coherence, for which 
the free states are the incoherent states, that are 
defined with respect to a specific basis 
$\{ | i \rangle \}$ of the considered Hilbert space 
${\cal H}_d$. A particularly interesting coherence 
measure is the relative entropy of coherence, which 
can be cast into the form
$R_d (\rho_\mathrm{A}) = -\sum_{i=1}^d 
p_i \log p_i-S_{vN} (\rho_\mathrm{A})$,
where 
$p_i=\langle i | \rho_\mathrm{A} | i \rangle$ 
\cite{BCP}. It is clearly lower than 
$\log d-S_{vN} (\rho_\mathrm{A})$. 
For a given density matrix 
$\rho_\mathrm{A}$, this value can be reached 
by performing unitary transformations \cite{YDGLS}, 
and hence $R_d^{sup}=\log d$. 
The situation is thus similar to that of 
the nonuniformity measure, based on $S_{vN}$, 
discussed above. So, $R_d$ obeys inequality 
\eqref{mi}, for any value of $d$, with familiar 
entanglement monotones. All the coherence monotones 
built with the help of a contractive distance \cite{BCP}, 
e.g., the relative entropy of coherence, satisfy 
$R_d\left (\{ U| i \rangle \},
U\rho_\mathrm{A}U^\dag \right)
=R_d(\{ | i \rangle \},\rho_\mathrm{A})$, 
for any unitary operator $U$, where we have 
denoted explicitly the dependence on the basis, 
with respect to which the incoherent states are 
defined. For such measures, the function $S$, 
given by eq.\eqref{S} and eq.\eqref{K}, and 
the supremum $R_d^{sup}$, do not depend on 
the basis $\{ | i \rangle \}$. This can be valid also 
for an entanglement monotone 
obeying eq.\eqref{SE} with $S$. An example is 
given by the definition \eqref{E}. In this case, 
the upper bound, $R_d^{sup}-E(\rho)$, 
to the coherence 
$R_d(\{ | i \rangle \},\rho_\mathrm{A})$, 
is independent of $\{ | i \rangle \}$. Moreover, 
for a pure state $\rho$, this bound is reached for 
some bases.

\begin{figure}
\centering 
\includegraphics[width=0.45\textwidth]{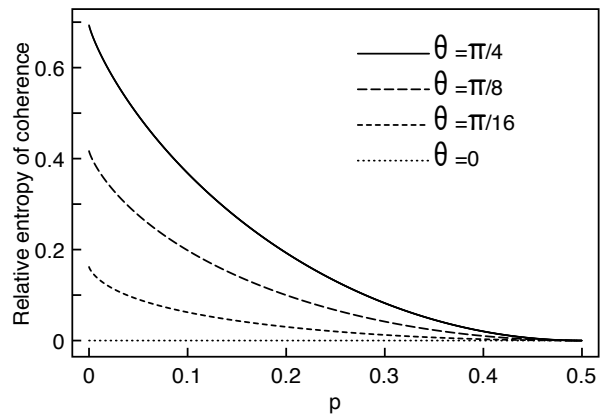}
\caption{\label{fig:Rec} Relative entropy of coherence 
for the long-time local state $p|0\rangle\langle 0| 
+(1-p)|1\rangle\langle 1|$, discussed in the text, 
and the basis $\{ \cos \theta |0\rangle + 
\sin \theta |1\rangle, \sin \theta |0\rangle 
- \cos \theta |1\rangle \}$, as a function of $p$, 
for four values of $\theta$. The coherence for 
$\theta=\pi/4$, is $\log 2-h(p)$, where $h(p)$ is 
the entanglement of formation of the global 
state. It is the bound given by eq.\eqref{mi}.}
\end{figure}

Using the above results, the role played by 
the entanglement of a system with its environment, 
in its decoherence, can be clarified. 
The coherence $R_d$ may vanish at long times, 
for a particular basis $\{ | i \rangle \}$, whereas 
the entanglement with the environment, allows 
nonzero coherence for other bases. Consider, 
for instance, that A is a two-level system, which 
interacts with a large system, B. For a pure 
dephasing Hamiltonian, and if A and B are initially 
in pure states, their common pure state reads
$|\psi\rangle=\sqrt{p} | 0 \rangle 
| \tilde 0 \rangle 
+ \sqrt{1-p} | 1 \rangle | \tilde 1 \rangle$, 
where $p \in [0,1]$, $\{ | 0 \rangle, | 1 \rangle \}$ 
is a basis of ${\cal H}_2$, and the states 
$|\tilde \imath \rangle$ are such that 
$| \tilde 0 \rangle=| \tilde 1 \rangle$ 
at initial time, and $\langle \tilde 0 | \tilde 1 \rangle$ 
goes to zero at long times \cite{PRA2,E}. 
In this long time regime, 
the relative entropy of coherence, for 
the basis $\{ | 0 \rangle, | 1 \rangle \}$, and 
the state $\rho_\mathrm{A}=
\operatorname{tr}_\mathrm{B}
|\psi \rangle \langle \psi |$, vanishes, and 
the entanglement of formation, 
for $\rho=|\psi \rangle \langle \psi |$, reaches 
$h(p) \equiv -p\log p-(1-p)\log (1-p)$. 
Inequality \eqref{mi} implies only that the relative 
entropy of coherence cannot exceed $\log 2-h(p)$, 
which is not zero if $p \ne 1/2$. Since $\rho$ is pure, 
this bound is attained for some bases, e.g., 
$\{ (| 0 \rangle+| 1 \rangle)/\sqrt{2}, 
(| 0 \rangle-| 1 \rangle)/\sqrt{2} \}$. Figure 
\ref{fig:Rec} shows the relative entropy of
 coherence for $\rho_\mathrm{A}$, and different 
bases.

In summary, we have derived a monogamy inequality 
for any local quantum resource and entanglement. 
We have shown that there is always a convex 
measure for a quantum resource, and that, for such 
a measure, there is a concave entropy, which satisfies 
a simple inequality with it. The monogamy inequality 
then ensues from the existence, for any concave 
entropy, of a bipartite entanglement monotone, for 
which the entanglement of the global state is 
necessarily lower than the entropies of the local 
states \cite{aX}. This inequality has been discussed 
for three local resources. It shows that 
the entanglement between parts of a system, and 
the entanglement between this system and another 
one, constrain each other. This entanglement 
monogamy is different from that usually considered 
\cite{CKW,KW,OV,BXW}. As seen, for three 
two-level systems in a pure symmetric state, 
this difference is manifest. For nonuniformity and 
coherence, the inequality can be written in terms 
of known resource measures \cite{GMNSH,BCP}, 
and entanglement monotones, such as 
the entanglement of formation \cite{HHHH}. 
For a large class of coherence monotones, to which 
belong the familiar ones \cite{BCP}, it gives an upper 
bound to the local coherence, which is independent 
of the basis with respect to which the coherence is 
evaluated. This bound is reached for some bases, 
when the global state is pure. Due to its generality, 
we expect the found monogamy inequality to have 
other consequences, for the quantum resources 
considered here, or for other ones.

\setcounter{equation}{0}

\section*{Supplemental Material}

In this supplemental material, we show that, 
for any quantum resource, there is a convex measure 
of it (proposition \ref{conv}), that the function $G$, 
defined by eq.\eqref{Kp}, is Schur-convex and convex 
(proposition \ref{propK}), and that there exist states 
for which the two sides of the monogamy inequality
$$R_{\boldsymbol d}(\rho_\mathrm{A})+E(\rho)
\le R_{\boldsymbol d}^{sup},$$ are equal, or 
infinitely close to each other (proposition \ref{eq}). 
We also evaluate $g(x)=co(h)(2x+1)$, where 
$co(h)$ is the convex hull of $h$, defined, 
on $[1,d^*]$, by
\begin{equation}
h(y)\equiv \tilde E_{max}-
\sup_{\{p_i \} \in {\cal F}(y)} G
\left( \sum_{i=1}^{d^*} p_i 
| i \rangle \langle i | \right) , \label{hp}
\end{equation}
with $\{ | i \rangle \}_{i=1}^{d^*}$ any basis, 
and ${\cal F}(y)$ the set of 
$\{p_i \}_{i=1}^{d^*}$ such that 
$p_i \ge 0$, $\sum_{i=1}^{d^*} p_i=1$, and 
$\sum_{i=1}^{d^*} \sqrt{p_i}=\sqrt{y}$, 
when system A consists of two two-level systems, 
and $\tilde E$ is the negativity. 

A resource measure $\{ R_{\boldsymbol d} \}$ 
obeys
\begin{equation}
R_{\boldsymbol d'}
\left(\sum_q K^{\phantom \dag}_q 
\rho_\mathrm{A} K^\dag_q \right)
\le R_{\boldsymbol d}(\rho_\mathrm{A}) , 
\label{mio1p} 
\end{equation}
where $K_q$ are the Kraus operators of 
the considered free operation, which are such 
that $\sum_q K^\dag_qK^{\phantom \dag}_q$ 
is equal to the identity operator on 
${\cal H}_{\boldsymbol d}$, or
\begin{equation}
\sum_q p_q R_{\boldsymbol d_q} 
\left(\rho_\mathrm{A}^{(q)}\right)
\le R_{\boldsymbol d}(\rho_\mathrm{A}), 
\label{mio2p}
\end{equation}
where $p_q=\operatorname{tr} 
(K^\dag_qK^{\phantom\dag}_q 
\rho_\mathrm{A})$ is the probability of 
outcome $q$, $\rho_\mathrm{A}^{(q)}=
K^{\phantom\dag}_q \rho_\mathrm{A} 
K^\dag_q/p_q$ is the corresponding state, and 
the sum runs over $q$ such that $p_q>0$. 
\begin{prop}\label{conv}
Consider non-negative functions 
$R_{\boldsymbol d}$ that vanish for free states, 
and satisfy eq.\eqref{mio1p} or eq.\eqref{mio2p}, 
with Kraus operators $K_q$.

There are non-negative convex functions 
$R^{ch}_{\boldsymbol d}$, independent 
of $\{ K_q \}$, that vanish for free states, obey 
eq.\eqref{mio1p} with $\{ K_q \}$, if 
$\boldsymbol d_q=\boldsymbol d'$, and fulfill 
eq.\eqref{mio2p} with $\{ K_q \}$, 
if $\{ R_{\boldsymbol d} \}$ does.
\end{prop} 
\begin{proof}
Since $R_{\boldsymbol d} \ge 0$, it has a convex 
hull, which is the maximum of the convex functions 
not larger than $R_{\boldsymbol d}$ \cite{HULp}, 
and is thus non-negative. We define 
$R^{ch}_{\boldsymbol d}$ as this convex hull. As 
$0 \le R^{ch}_{\boldsymbol d} 
\le R_{\boldsymbol d}$, $R^{ch}_{\boldsymbol d}$ 
vanishes whenever $R_{\boldsymbol d}$ does, 
e.g., for free states. 

Assume $R_{\boldsymbol d}$ and 
$R_{\boldsymbol d'}$ satisfy eq.\eqref{mio1p}, and 
define the function $H$ by $H(\rho_\mathrm{A}) 
\equiv R^{ch}_{\boldsymbol d'}
[\Phi (\rho_\mathrm{A})]$, 
where $\Phi (\rho_\mathrm{A}) \equiv \sum_q 
K^{\phantom\dag}_q \rho_\mathrm{A} K^\dag_q$. 
Due to $R^{ch}_{\boldsymbol d'} \le 
R^{\phantom c}_{\boldsymbol d'}$ 
and eq.\eqref{mio1p}, 
$H$ is not greater than $R_{\boldsymbol d}$. 
For the states $\rho'_\mathrm{A}$, 
$\rho''_\mathrm{A}$, and $\rho_\mathrm{A}
=\tau \rho'_\mathrm{A}
+\bar \tau \rho''_\mathrm{A}$, 
where $\bar \tau=1-\tau$, and $\tau \in [0,1]$, 
one obtains $$H(\rho_\mathrm{A})
=R^{ch}_{\boldsymbol d'}[
\tau \Phi (\rho'_\mathrm{A})
+\bar \tau \Phi (\rho''_\mathrm{A})] 
\le \tau H (\rho'_\mathrm{A})+\bar \tau 
H (\rho''_\mathrm{A}) ,$$ using the linearity 
of $\Phi$, and the convexity 
of $R^{ch}_{\boldsymbol d'}$. Since $H$ is 
convex and not larger than $R_{\boldsymbol d}$, 
$H \le R^{ch}_{\boldsymbol d}$, i.e., 
$R^{ch}_{\boldsymbol d}$ and 
$R^{ch}_{\boldsymbol d'}$ obey eq.\eqref{mio1p}.

Assume now that $\{ R_{\boldsymbol d} \}$ fulfills 
eq.\eqref{mio2p}, and define the functions 
$I_{\boldsymbol d}$ by $I_{\boldsymbol d}(\omega) 
\equiv p R^{ch}_{\boldsymbol d}(\omega /p)$, 
where $\omega$ is any positive Hermitian operator 
on ${\cal H}_{\boldsymbol d}$, of trace 
$p=\operatorname{tr} \omega >0$, 
and $I_{d}(0) \equiv 0$. For $\omega'$, $\omega''$, 
and $\omega=\tau \omega'+\bar \tau \omega''$, 
where $\bar \tau=1-\tau$, and $\tau \in [0,1]$, 
the convexity of $R^{ch}_{\boldsymbol d}$ leads 
to $$I_{\boldsymbol d}(\omega) 
= p R_{\boldsymbol d}^{ch} 
\left( \tau \frac{p'}{p} \frac{\omega'}{p'} 
+ \bar \tau \frac{p''}{p} \frac{\omega''}{p''} \right)
\le \tau I_{\boldsymbol d}(\omega')
+\bar \tau I_{\boldsymbol d}(\omega'') ,$$
where $p=\operatorname{tr} \omega$, and $p'$ 
and $p''$ are given by similar expressions. 
Consequently, the function $J$, defined by 
$J(\rho_\mathrm{A}) \equiv 
\sum_q I_{\boldsymbol d_q}(K^{\phantom\dag}_q 
\rho_\mathrm{A} K^\dag_q)$, is convex. Moreover, 
due to $R^{ch}_{\boldsymbol d_q} 
\le R_{\boldsymbol d_q}$ and eq.\eqref{mio2p}, 
$J$ is not greater than $R_{\boldsymbol d}$. 
Thus, $J \le R^{ch}_{\boldsymbol d}$, i.e., 
$\{ R^{ch}_{\boldsymbol d} \}$ obeys 
eq.\eqref{mio2p}. 
For $\boldsymbol d_q=\boldsymbol d'$, 
since $R^{ch}_{\boldsymbol d'}$ is convex, 
$R^{ch}_{\boldsymbol d'}
[\Phi (\rho_\mathrm{A})]
\le J(\rho_\mathrm{A})$, and hence 
$R^{ch}_{\boldsymbol d}$ and 
$R^{ch}_{\boldsymbol d'}$ fulfill eq.\eqref{mio1p}.
\end{proof} 

From an arbitrary convex function 
$R_{\boldsymbol d}$, of the density operators on 
${\cal H}_{\boldsymbol d}$, non-negative and 
bounded, we define, 
for any state $\rho_\mathrm{A}$ of any finite 
system A,
\begin{equation}
G(\rho_\mathrm{A}) \equiv 
\sup_{\{ | i \rangle \}} R_{\boldsymbol d}
\left( \sum_{i=1}^r \lambda_i(\rho_\mathrm{A}) 
| i \rangle \langle i | \right) , \label{Kp}
\end{equation}
for $r \le n(\boldsymbol d)$, and $0$ otherwise, 
where $n(\boldsymbol d) \equiv d_1 d_2 \ldots$ 
is the dimension of ${\cal H}_{\boldsymbol d}$, 
$r$ is the rank of $\rho_\mathrm{A}$, and 
$\lambda_i(M)$ denotes the eigenvalues of 
the Hermitian operator $M$, with 
$\lambda_i(M) \ge \lambda_{i+1}(M)$.
The supremum is taken over the bases 
$\{ | i \rangle \}$ of ${\cal H}_{\boldsymbol d}$. 
\begin{prop}\label{propK}
The function \eqref{Kp} is
\begin{enumerate}[label=\roman*),
leftmargin=15pt,itemsep=-2pt,topsep=3pt]
\item Schur-convex, i.e., \label{pointii}
$G(\rho_\mathrm{A}) \le G(\rho'_\mathrm{A})$ 
when $\rho'_\mathrm{A}$ majorizes 
$\rho_\mathrm{A}$, where $\rho_\mathrm{A}$ 
and $\rho'_\mathrm{A}$ are states of any finite 
systems,
\item convex, i.e., \label{pointiii}
$G(\tau \rho_\mathrm{A}+ 
\bar \tau \rho'_\mathrm{A}) 
\le \tau G(\rho_\mathrm{A}) 
+ \bar \tau G(\rho'_\mathrm{A})$, 
where $\rho_\mathrm{A}$ and 
$\rho'_\mathrm{A}$ are states of a same system, 
$\bar \tau=1-\tau$, and $\tau \in [0,1]$.
\end{enumerate}
\end{prop}
\begin{proof}
Here, we denote $n(\boldsymbol d)$ by $n$. 
For a density operator $\rho_\mathrm{A}$ of rank 
$r \le n$, we rewrite eq.\eqref{Kp} as 
\begin{equation}
G(\rho_\mathrm{A})=f\left[\boldsymbol 
\lambda_n(\rho_\mathrm{A})\right] ,\label{Kf}
\end{equation}
where the $n$-component vector 
$\boldsymbol \lambda_n(\rho_\mathrm{A})$
is made up of the $r$ nonvanishing eigenvalues 
$\lambda_i(\rho_\mathrm{A})$, in decreasing order, 
followed by $n-r$ zeros. The function $f$ of 
the $n$-component probability vectors 
$\boldsymbol p$, i.e., such that $p_i \ge 0$ 
and $\sum_{i=1}^n p_i=1$, is given by
\begin{equation}
f(\boldsymbol p) \equiv  \sup_{\{|i\rangle \}}
R_{\boldsymbol d}\left(
\sum_{i=1}^n p_i | i \rangle \langle i | \right)  , 
\label{f}
\end{equation}
where the supremum is taken over the bases 
$\{|i\rangle \}$ of ${\cal H}_{\boldsymbol d}$. 
It is clear, from its definition, that $f$ is 
a symmetric function of the components $p_i$. 
Consider the probability vectors $\boldsymbol p$, 
$\boldsymbol p'$, and $\boldsymbol p'' \equiv 
\tau \boldsymbol p+\bar \tau \boldsymbol p'$, 
where $\bar \tau=1-\tau$, and $\tau \in [0,1]$. 
The convexity of $R_{\boldsymbol d}$ and 
the definition \eqref{f}, give, for any basis 
$\{ | i \rangle \}$ of ${\cal H}_{\boldsymbol d}$, 
$R_{\boldsymbol d}( 
\sum_{i=1}^n p''_i | i \rangle \langle i |) 
\le  \tau f(\boldsymbol p)+\bar \tau f(\boldsymbol p')$, 
which leads to the convexity of $f$. Being symmetric 
and convex, $f$ is Schur-convex \cite{MOAp}.

\ref{pointii} Consider two density operators 
$\rho_\mathrm{A}$ and $\rho'_\mathrm{A}$, 
of ranks $r$ and $r'$, respectively, such that 
$\rho'_\mathrm{A}$ majorizes $\rho_\mathrm{A}$. 
Thus, $r' \le r$. If $r>n$, the inequality 
$G(\rho_\mathrm{A}) \le G(\rho'_\mathrm{A})$ is 
trivially obeyed. If $r \le n$, $G(\rho_\mathrm{A})$ 
and $G(\rho'_\mathrm{A})$ are both given by 
eq.\eqref{Kf}.  Since $\boldsymbol\lambda_n
(\rho'_\mathrm{A})$ majorizes 
$\boldsymbol\lambda_n(\rho_\mathrm{A})$, 
and $f$ is Schur-convex, 
$G(\rho_\mathrm{A}) \le G(\rho'_\mathrm{A})$.

\ref{pointiii} Consider the states $\rho_\mathrm{A}$ 
and $\rho'_\mathrm{A}$, of a same system, 
of ranks $r$ and $r'$, respectively, 
and $\rho''_\mathrm{A} \equiv 
\tau \rho_\mathrm{A}
+\bar \tau \rho'_\mathrm{A}$, where 
$\bar \tau=1-\tau$, and $\tau \in ]0,1[$. 
We assume, without loss of generality, 
that $r'\le r$. Due to Ky Fan eigenvalue inequality, 
$\tau \boldsymbol \lambda(\rho_\mathrm{A})
+\bar \tau \boldsymbol \lambda(\rho'_\mathrm{A})$ 
majorizes $\boldsymbol \lambda(\rho''_\mathrm{A})$, 
where $\boldsymbol \lambda(M)$ is the vector made 
up of the eigenvalues $\lambda_i(M)$, in decreasing 
order \cite{MOAp,Fp}. The rank $r''$ of 
$\rho''_\mathrm{A}$, is hence not smaller than $r$. 
The convexity inequality for $G$, is obviously satisfied 
with $\rho_\mathrm{A}$ and $\rho'_\mathrm{A}$, 
if $r''>n$. In the case $r'' \le n$, 
$G(\rho_\mathrm{A})$, $G(\rho'_\mathrm{A})$, 
and $G(\rho''_\mathrm{A})$ are given by 
eq.\eqref{Kf}. Moreover, $\tau \boldsymbol p
+\bar \tau \boldsymbol p'$ majorizes 
$\boldsymbol p''$, where $\boldsymbol p
=\boldsymbol\lambda_n(\rho_\mathrm{A})$, 
and $\boldsymbol p'$ and $\boldsymbol p''$ 
are given by similar expressions, which leads to 
$$f(\boldsymbol p'') \le 
f(\tau \boldsymbol p+\bar \tau\boldsymbol p') 
\le \tau f(\boldsymbol p)
+\bar \tau f(\boldsymbol p') , $$ 
since $f$ is  Schur-convex and convex. 
\end{proof}

\begin{prop}\label{eq}
Consider a system A whose Hilbert space 
is ${\cal H}_{\boldsymbol d}$, a $m$-level 
system B, $r$ probabilities $p_i$, such that 
$\sum_{i=1}^r p_i=1$, and $p_i \ge p_{i+1}$, 
with $r \le m,n(\boldsymbol d)$, 
and an entanglement monotone $E$, related to 
$R_{\boldsymbol d}$, by
\begin{equation}
 R_{\boldsymbol d}^{sup}-G(\rho_\mathrm{A})= 
\max_{\rho : \operatorname{tr}_\mathrm{B} \rho 
= \rho_\mathrm{A}} E(\rho) . \nonumber
\end{equation}

For any $\epsilon>0$, there are pure states 
$\rho$ of A and B, such that 
$\lambda_i(\rho_\mathrm{A})=p_i$ for $i \le r$, 
and $0$ otherwise, and $R_{\boldsymbol d}^{sup}
-R_{\boldsymbol d}(\rho_\mathrm{A})
-E(\rho) < \epsilon$.
\end{prop}
\begin{proof}
Consider any pure state 
$\tilde \rho=|\psi \rangle \langle \psi |$ of A and B,
where $|\psi \rangle$ has Schmidt coefficients 
$\sqrt{p_i}$, and define $\cal F$ the set of all pure 
states $\rho=U\tilde \rho U^\dag$, where $U$ is 
any unitary operator of A. For any 
$\rho \in \cal F$, $E(\rho)=E(\tilde \rho)$, 
and $\lambda_i(\rho_\mathrm{A})= p_i$ for $i \le r$, 
and $0$ otherwise. Since $\tilde \rho$ is pure, 
$E(\tilde \rho)=R^{sup}_{\boldsymbol d}
-\sup_U R_{\boldsymbol d}
(\operatorname{tr}_B U\tilde \rho U^\dag)$, 
where the supremum is taken over the unitary 
operators $U$ of A. Thus, 
$\sup_{\rho \in \cal F} 
[ E(\rho)+ R_{\boldsymbol d}(\rho_\mathrm{A})]
=R_{\boldsymbol d}^{sup}$, which finishes 
the proof.
\end{proof}

When system A consists of two two-level systems, 
and $\tilde E$ is the negativity, 
$G(\rho_\mathrm{A})
=\max\{0,H(\rho_\mathrm{A})\}/2$, where
\begin{equation}
H(\rho_\mathrm{A})= \sqrt{(p_1 - p_3)^2
+(p_2 - p_4)^2}-p_2 - p_4  , \nonumber
\end{equation}
with $p_i=\lambda_i(\rho_\mathrm{A})$
\cite{VADMp}, and hence $\tilde E_{max}=1/2$. 
For $d^*=2$, in eq.\eqref{hp}, the set ${\cal F}(y)$ 
has only one element $\{p_1,p_2\}$. 
Consequently, $g$ is given by 
$$g(x)=\Big(3/2-\sqrt{1-2x^2}-\sqrt{1/4-x^2}\Big)/2 ,$$
which increases from $0$ to $(3-\sqrt{2})/4$, as $x$ 
varies from $0$ to $1/2$. 
For $d^*=3$, to evaluate $h(y)$, we first determine 
the maximum of $H(\rho_\mathrm{A})$, 
for $\{p_1,p_2,p_3 \} \in {\cal F}(y)$. 
It is reached for $p_2=p_3$, and decreases from $1$ 
to $0$, as $y$ varies from $1$ to $3$. It leads to 
$$g(x)=1/2-k(x)\Big(\sqrt{1+[9/k(x)-3]^2}-1\Big)/18,$$ 
where $k(x)=(\sqrt{2x+1}-\sqrt{1-x})^2$. This function 
increases from $0$ to $1/2$, as $x$ varies from $0$ to $1$. 
For $d^*=4$, the maximum of $H(\rho_\mathrm{A})$, 
is reached for $p_2=p_3=p_4$, and is equal to 
$f(y)\equiv(\sqrt{y}+\sqrt{12-3y})^2/8-1$. For 
$y\le 2+\sqrt{3}$, it is positive, and hence 
$h(y)=1/2-f(y)/2$. For $y\ge 2+\sqrt{3}$, it is 
negative, and hence $h(y)=1/2$. The convex hull 
$co(h)$ is equal to $h$, 
for $y\le y_0$, and to $a(y-4)+1/2$, for $y\ge y_0$, 
where $y_0$ and $a$ are determined by 
$h'(y_0)=a$, and $h(y_0)=a(y_0-4)+1/2$, with 
$h'$ the derivative of $h$. These conditions give 
$y_0=3$ and $a=1/4$, and hence
\begin{eqnarray}
g(x) &=& 
1-\left(\sqrt{1+2x}+\sqrt{9-6x} \right)^2/16
\;\mbox{for}\; x\in[0,1] , \nonumber \\
&=&  x/2 -1/4\;\mbox{for}\; x\in[1,3/2] ,
\nonumber
\end{eqnarray}
which increases from $0$ to $1/2$, as $x$ 
varies from $0$ to $3/2$.

\end{document}